\documentclass[11pt]{article}
\usepackage{geometry}
\usepackage{amsmath, amssymb, amsthm}
\usepackage{bm}
\usepackage{fullpage}
\newtheorem{theorem}{Theorem}
\newtheorem{lemma}{Lemma}

\title{The Set of Stable Matchings and the Core in a Matching Market with Ties and Matroid Constraints\thanks{%
This work was supported by JST ERATO Grant Number JPMJER2301, Japan.}}
\author{Naoyuki Kamiyama}
\date{Institute of Mathematics for Industry \\ Kyushu University \\ Fukuoka, Japan \\
{\ttfamily kamiyama@imi.kyushu-u.ac.jp}}

\begin{document}

\maketitle

\begin{abstract}
In this paper, we consider a many-to-one matching market where 
ties in the preferences of agents are allowed. 
For this market with capacity constraints, 
Bonifacio, Juarez, Neme, and Oviedo
proved some relationship between the set of 
stable matchings and the core.
In this paper, we consider 
a matroid constraint that is a generalization 
of a capacity constraint. 
We prove that the results proved by 
Bonifacio, Juarez, Neme, and Oviedo 
can be generalized to this setting.
\end{abstract} 

\section{Introduction}

Stability, 
which was introduced by
Gale and Shapley~\cite{GaleS62}, is 
one of the most central properties in a matching market. 
This property guarantees that 
there does not exist an unmatched pair such that 
both agents in this pair have incentive to deviate 
from the current matching. Furthermore, the core
is another important set of solutions in a matching market. 
We say that a matching $\mu$ is dominated by another matching $\sigma$
if there exists a coalition $C$ of agents such that 
(i) in $\sigma$, the agents in $C$
are matched to agents in $C$, and
(ii) 
the agents in $C$ prefer $\sigma$ to $\mu$.
Then the core is the set of matchings 
that are not dominated by any other matching. 
It is known that 
the set of stable matchings 
coincides with the core 
in a one-to-one matching market with strict preferences
(see, e.g., \cite[Appendix~2]{Roth85}). 

In this paper, we consider a matching market where 
ties in the preferences of agents are allowed, i.e., 
an agent may be indifferent between potential partners. 
For this matching market, 
Irving~\cite{Irving94}
introduced the following three definitions of stability. 
The first one is weak stability\footnote{In \cite{BonifacioJNO24}, 
a weakly stable matching in this paper is  
simply called a stable matching.}.  
This guarantees that 
there does not exist an unmatched pair of agents such that 
both agents prefer the other agent in this pair 
to the current partner. 
Irving~\cite{Irving94} proved 
that there always exists a weakly stable matching and
we can find a weakly stable matching 
in polynomial time by slightly modifying the algorithm 
of Gale and Shapley~\cite{GaleS62} for strict preferences. 
The second one is strong stability.
This guarantees that 
there does not exist an unmatched pair of agents such that
both agents weakly prefer\footnote{%
This means that both agent prefers the other agent in this pair to the current partner, or 
they are indifferent between the other agent in this pair and the current partner.} 
the other agent in this pair 
to the current partner, and 
at least one of the agents prefers the other agent in this pair 
to the current partner.
The last one is super-stability.
This guarantees
that there does not exist an unmatched pair of agents such that 
both agents weakly prefer the other agent in this pair 
to the current partner.
It is known that 
a super-stable matching and a strongly stable matching 
may not exist~\cite{Irving94}. 
Irving~\cite{Irving94} proved that,
in the one-to-one setting, 
we can check the existence of a strongly stable matching and 
a super-stable matching in polynomial time
(see also \cite{Manlove99}). 
Irving, Manlove, and Scott~\cite{IrvingMS00}
proved that, in the many-to-one setting
with capacity constraints, we can check the existence of 
a super-stable matching in polynomial time.
Furthermore,  
Irving, Manlove, and Scott~\cite{IrvigMS03}
and Kavitha, Mehlhorn, Michail, and Paluch~\cite{KavithaMMP07}
proved that, in the many-to-one setting
with capacity constraints, 
we can check the existence of 
a strongly stable matching in polynomial time. 

Bonifacio, Juarez, Neme, and Oviedo~\cite{BonifacioJNO24}
considered the relationship between the set of stable matchings 
and the core in a many-to-one matching market with ties 
and capacity constraints.
Bonifacio, Juarez, Neme, and Oviedo~\cite{BonifacioJNO24}
first introduced the following three definitions of the core. 
The first one is the weak core\footnote{%
In \cite{BonifacioJNO24}, the weak core in this paper is 
simply called the core.}.
This is the set of matchings such that 
there does not exist a coalition of agents such that the agents 
in this coalition prefer the new matching.
The second one is the strong core.
This is the set of matchings such that 
there does not exist a coalition of agents such that the agents 
in this coalition weakly prefer the new matching, and 
at least one of the agents in this coalition prefers the new matching.
The last one is the super core.
This is the set of matchings such that 
there does not exist a coalition of agents such that the agents 
in this coalition weakly prefer the new matching, and 
at least one of the agents in this coalition is matched to different 
agents in the new matching.
Then Bonifacio, Juarez, Neme, and Oviedo~\cite{BonifacioJNO24}
proved that, under the assumption that the preferences are responsive,  
the set of weakly 
stable matchings is contained in the weak core, 
the set of strongly stable matchings coincides with 
the strong core, and 
the set of super-stable matchings coincides with the 
super core. 

In this paper, we consider the relationship between the set of stable matchings and 
the core in a matching market with 
ties and matroid constraints.
Matroids can express not only capacity constraints 
but also more complex constraints including hierarchical 
capacity constraints.
Thus, matroid constraints are important from not only the theoretical 
viewpoint but also the practical viewpoint, and 
the stable matching problem with matroid constraints
has been actively 
studied (see, e.g., 
\cite{Fleiner03,FleinerK16,Kamiyama15,Kamiyama19,Kamiyama20,Kamiyama22,Kamiyama22b}).
In this paper, we prove that the results proved in \cite{BonifacioJNO24} 
can be generalized to this setting.

\section{Preliminaries} 

For each set $X$ and each element $x$, we define 
$X + x := X \cup \{x\}$ and 
$X - x := X \setminus \{x\}$. 
For each set $X$, we define 
$X - \emptyset := X$. 
For each positive integer $z$, we define $[z] := \{1,2,\dots,z\}$.

\subsection{Matroids} 

A pair ${\bf M}= (U,\mathcal{I})$ 
of a finite set $U$ and a non-empty family $\mathcal{I}$
of subsets of $U$ 
is called a \emph{matroid} if, for
every pair of subsets $I,J \subseteq U$, the following conditions are satisfied. 
\begin{description}
\item[(I1)]
If $I \subseteq J$ and $J \in \mathcal{I}$, then
$I \in \mathcal{I}$.  
\item[(I2)]
If $I, J \in \mathcal{I}$ and 
$|I| < |J|$, then 
there exists an element $u \in J \setminus I$ such that 
$I + u \in \mathcal{I}$. 
\end{description} 

Let ${\bf M} = (U, \mathcal{I})$ be a matroid. 
Then 
an element in $\mathcal{I}$ is called an 
\emph{independent set of ${\bf M}$}.
An inclusion-wise maximal independent set of ${\bf M}$ 
is called a \emph{base of ${\bf M}$}. 
Then (I2) implies that 
all the bases of ${\bf M}$ have the same size. 
For every independent set $I$ of ${\bf M}$, 
if there exists a base $B$ of ${\bf M}$ such that 
$|I| = |B|$, then $I$ is a base of ${\bf M}$. 
A subset $X \subseteq U$ 
such that $X \notin \mathcal{I}$ 
is called a \emph{dependent set of ${\bf M}$}. 
An inclusion-wise minimal dependent set of ${\bf M}$ 
is called a \emph{circuit of ${\bf M}$}. 

\begin{lemma}[{See, e.g., \cite[Lemma 1.1.3]{Oxley11}}] \label{lemma:Oxley11:lemma_1_1_3}
Let ${\bf M}$ be a matroid, and 
let $C_1,C_2$ be distinct circuits of ${\bf M}$ such that 
$C_1 \cap C_2 \neq \emptyset$. 
Then for every element $u \in C_1 \cap C_2$, 
there exists a circuit $C$ of ${\bf M}$ 
such that $C \subseteq (C_1 \cup C_2) - u$.
\end{lemma}

Let ${\bf M} = (U, \mathcal{I})$ be a matroid,
and let $I$ be an independent set of ${\bf M}$.  
Then for every element $u \in U \setminus I$ such that 
$I + u \notin \mathcal{I}$, 
(I1) implies that 
$I + u$ contains a circuit $C$ of ${\bf M}$ as a subset, and 
$u \in C$.
It follows from Lemma~\ref{lemma:Oxley11:lemma_1_1_3} that 
such a circuit is uniquely determined.
This circuit is called the  
\emph{fundamental circuit of $u$ with respect to $I$ and ${\bf M}$}. 
For each element $u \in U \setminus I$ such that 
$I + u \notin \mathcal{I}$,
let ${\sf C}_{\bf M}(u,I)$ denote  
the fundamental circuit of $u$ with respect to $I$ and ${\bf M}$. 
It is known that, 
for every element $u \in U \setminus I$ such that 
$I + u \notin \mathcal{I}$, 
${\sf C}_{{\bf M}}(u,I)$ is the set of 
elements $v \in I  + u$ such that  
$I + u - v \in \mathcal{I}$
(see, e.g., \cite[p.20, Exercise 5]{Oxley11}). 
For each element $u \in U \setminus I$ such that 
$I + u \notin \mathcal{I}$,
we define ${\sf D}_{{\bf M}}(u,I) := {\sf C}_{{\bf M}}(u,I) - u$. 

Let ${\bf M} = (U,\mathcal{I})$ be a matroid.
For each subset $X \subseteq U$, 
we define $\mathcal{I}|X := \{I \in \mathcal{I} \mid I \subseteq X\}$ 
and 
${\bf M}|X := (X, \mathcal{I}|X)$. 
Then it is known that 
${\bf M}|X$ is a matroid
for 
every subset $X \subseteq U$ (see, e.g., \cite[p.20]{Oxley11}).

\subsection{Setting} 

Throughout this paper, we are given a finite simple undirected bipartite graph 
$G = (V,E)$. 
We assume that 
the vertex set $V$ of $G$ is partitioned into 
$D, H$, and every edge in the edge set $E$ of $G$ connects a vertex in 
$D$ and a vertex in $H$.
We call a vertex in $D$ (resp.\ $H$) a 
\emph{doctor} (resp.\ \emph{hospital}). 
For each doctor $d \in D$ and each hospital $h \in H$, if there exists an edge 
connecting $d$ and $h$, then $(d,h)$ denotes this edge. 
For each subset $F \subseteq E$ and each doctor $d \in D$ (resp.\ hospital 
$h \in H$), 
we define $F(d)$ (resp.\ $F(h)$) as the set of edges $(d^{\prime},h^{\prime}) \in F$
such that $d^{\prime} = d$ (resp.\ $h^{\prime} = h$).  

For each hospital $h \in H$, we are given a matroid 
${\bf M}_h = (E(h),\mathcal{I}_h)$ such that, for every 
edge $e \in E(h)$, we have $\{e\} \in \mathcal{I}_h$.   
For each doctor $d \in D$, we are given a transitive 
and complete\footnote{%
For every pair of elements $e,f \in E(d) \cup \{\emptyset\}$, 
at least one of $e \succsim_d f$ and $f \succsim_d e$ holds.} 
binary relation $\succsim_d$ on $E(d) \cup \{\emptyset\}$ 
such that, for every edge $e \in E(d)$, 
we have $e \succsim_d \emptyset$ and $\emptyset \not\succsim_d e$. 
For each hospital $h \in H$, we are given a 
transitive and complete binary relation $\succsim_h$ on $\mathcal{I}_h$ such that, 
for every independent set $I \in \mathcal{I}_h$, we have 
$I \succsim_h \emptyset$ and $\emptyset \not\succsim_h I$.
For each doctor $d \in D$ and 
each pair of elements $e,f \in E(d) \cup \{\emptyset\}$, 
we write 
$e \succ_d f$ (resp.\ $e \sim_d f$) 
if $e \succsim_d f$ and $f \not\succsim_d e$
(resp.\ $e \succsim_d f$ and $f \succsim_d e$). 
For each hospital $h \in H$ and 
each pair of independent sets $I,J \in \mathcal{I}_h$, 
we write 
$I \succ_h J$ (resp.\ $I \sim_h J$) 
if $I \succsim_h J$ and $J \not\succsim_h I$
(resp.\ $I \succsim_h J$ and $J \succsim_h I$). 
Furthermore, for each hospital $h \in H$ and 
each pair of edges $e,f \in E(h)$, 
we write $e \succsim_h f$, 
$e\succ_h f$, and $e \sim_h f$ 
instead of $\{e\} \succsim_h \{f\}$, 
$\{e\} \succ_h \{f\}$, and 
$\{e\} \sim_h \{f\}$, respectively. 

In this paper, we assume that, for every hospital $h \in H$, 
$\succsim_h$ is \emph{responsive}. 
More precisely, for every hospital $h \in H$, 
we assume that
$\succsim_h$ satisfies 
the following conditions. 
\begin{itemize}
\item
For every independent set $I \in \mathcal{I}_h$ 
and every edge $e \in I$, we have 
$I \succ_h I - e$. 
\item
For every independent set $I \in \mathcal{I}_h$ and every 
pair of edges $e \in E(h) \setminus I$ 
and $f \in I$ such that 
$I + e - f \in \mathcal{I}_h$,  
$I + e - f \succsim_h I$ if and only if $e \succsim_h f$. 
\end{itemize}

\begin{lemma}
Let $h$ be a hospital in $H$.
Then for every independent set $I \in \mathcal{I}_h$ and every 
pair of edges $e \in E(h) \setminus I$ 
and $f \in I$ such that 
$I + e - f \in \mathcal{I}_h$,  
$I + e - f \succ_h I$ {\normalfont (}resp.\ $I + e - f \sim_h I${\normalfont )} 
if and only if $e \succ_h f$ 
{\normalfont (}resp.\ $e \sim_h f${\normalfont )}. 
\end{lemma}
\begin{proof}
Define $J := I + e - f$. 
Then $J + f - e = I \in \mathcal{I}_h$.
Thus, the definition of $\succsim_h$ implies that 
$J + f - e \succsim_h J$ (i.e., 
$I \succsim_h I + e - f$) if and only if 
$f \succsim_h e$. 
This completes the proof.
\end{proof} 

A subset $\mu \subseteq E$ is called a \emph{matching in $G$} 
if the following conditions are satisfied. 
\begin{itemize}
\item
For every doctor $d \in D$, we have $|\mu(d)| \le 1$. 
\item
For every hospital $h \in H$, we have $\mu(h) \in \mathcal{I}_h$.
\end{itemize} 
For each matching $\mu$ in $G$ and each 
doctor $d \in D$ such that $\mu(d) \neq \emptyset$, 
we do not distinguish $\mu(d)$ and the 
unique element in $\mu(d)$. 

If, for each hospital $h \in H$, we are given a positive integer $c_h$ and 
we define $\mathcal{I}_h$ as the family of 
subsets $F \subseteq E(h)$ such that $|F| \le c_h$, then 
our setting coincides with the setting considered in \cite{BonifacioJNO24}. 
This kind of matroid is called a \emph{uniform matroid}.
In addition, we give another example of our model. 
For each hospital $h \in H$, we are given a  
family $\mathcal{P}_h$ of subsets of $E(h)$
such that, 
for every pair of distinct elements $X,Y \in \mathcal{P}_h$, 
$X \cap Y = \emptyset$ or  
$X \subseteq Y$ or $Y \subseteq X$
(i.e., $\mathcal{P}_h$ is a \emph{laminar} family). 
In addition, for each hospital $h \in H$ and each 
element $P \in \mathcal{P}_h$, we are given 
a positive integer $c_h(P)$.
Define 
$\mathcal{I}_h$ as the family of subsets $F \subseteq E(h)$ 
such that,
for every element $P \in \mathcal{P}_h$, 
$|F \cap P| \le c_ h(P)$. 
Then it is known that the ordered pair 
${\bf M}_h = (E(h), \mathcal{I}_h)$ defined in this way is a matroid. 
This kind of constraint was considered in, e.g., \cite{Huang10}. 

\subsection{Stable matchings} 

Let $\mu$ be a matching in $G$. 
Let $e = (d,h)$ be an edge in $E \setminus \mu$. 
Then we say that $e$ \emph{weakly blocks} 
(resp.\ \emph{strongly blocks}) 
$\mu$ on $d$ if 
$e \succsim_d \mu(d)$
(resp.\ $e \succ_d \mu(d)$). 
We say that $e$ \emph{weakly blocks} 
(resp.\ \emph{strongly blocks})
$\mu$ on $h$ if at least one of 
the following conditions is satisfied. 
\begin{itemize}
\item
$\mu(h) + e \in \mathcal{I}_h$.
\item
$\mu(h) + e \notin \mathcal{I}_h$, and 
there exists an edge $f \in {\sf D}_{{\bf M}_h}(e,\mu(h))$ such that 
$e \succsim_h f$ 
(resp.\ $e \succ_h f$). 
\end{itemize}
We say that 
$e$ \emph{strongly blocks} $\mu$ if 
$e$ strongly blocks $\mu$ on both $d$ and $h$. 
We say that 
$e$ \emph{weakly blocks} $\mu$ if 
$e$ weakly blocks $\mu$ on both $d$ and $h$, 
and $e$ strongly blocks $\mu$ on at least one of 
$d$ and $h$.
We say that 
$e$ \emph{super weakly blocks} $\mu$ if 
$e$ weakly blocks $\mu$ on both $d$ and $h$. 

Let $\mu$ be a matching in $G$.
Then $\mu$ is said to be 
\emph{weakly stable} if 
any edge in $E \setminus \mu$ does not strongly block $\mu$. 
Let $\mathcal{S}$ denote the set of weakly stable matchings in $G$.
Furthermore, $\mu$ is said to be 
\emph{strongly stable} if 
any edge in $E \setminus \mu$ does not weakly block $\mu$. 
Let $\mathcal{SS}$ denote the set of strongly stable matchings in $G$.
Finally, $\mu$ is said to be 
\emph{super-stable} if 
any edge in $E \setminus \mu$ does not super weakly block $\mu$. 
Let $\mathcal{SSS}$ denote the set of super-stable matchings in $G$.

\subsection{Core} 

For each matching $\mu$ in $G$ and each non-empty subset $C \subseteq V$, 
$\mu$ is said to be \emph{consistent with $C$} if, 
for every vertex $v \in C$, $\mu(v) \subseteq E(C)$, 
where $E(C)$ denotes the set of edges $(d,h) \in E$ 
such that $d,h \in C$. 
For each non-empty subset $C \subseteq V$, 
let $\mathcal{M}(C)$ denote the set of matchings in $G$ that are consistent 
with $C$. 

Let $\mu$ be a matching in $G$. 
Let $C$ be a non-empty subset of $V$.
We say that \emph{$C$ strongly blocks $\mu$} if there exists 
a matching $\sigma \in \mathcal{M}(C)$
such that,
for every vertex $v \in C$, 
$\sigma(v) \succ_v \mu(v)$. 
We say that \emph{$C$ weakly blocks $\mu$} if there exists 
a matching $\sigma \in \mathcal{M}(C)$ 
satisfying the following conditions. 
\begin{itemize}
\item
For every vertex $v \in C$, 
$\sigma(v) \succsim_v \mu(v)$. 
\item 
There exists a vertex $v \in C$
such that $\sigma(v) \succ_v \mu(v)$. 
\end{itemize}
Furthermore, we say that \emph{$C$ super weakly blocks $\mu$} if there exists 
a matching $\sigma \in \mathcal{M}(C)$ 
satisfying the following conditions. 
\begin{itemize}
\item
For every vertex $v \in C$, 
$\sigma(v) \succsim_v \mu(v)$. 
\item 
There exists a vertex $v \in C$
such that $\sigma(v) \neq \mu(v)$. 
\end{itemize}

\begin{lemma} \label{lemma:non-empty} 
Let $\mu$ be a matching $\mu$.
Let $C$ be a non-empty subset of $V$ 
that super weakly blocks $\mu$. 
Then $C \cap H \neq \emptyset$. 
\end{lemma}
\begin{proof}
Since $C$ super weakly blocks $\mu$, there exists 
a matching $\sigma \in \mathcal{M}(C)$
such that 
(i)  
$\sigma(v) \succsim_v \mu(v)$
for every vertex $v \in C$, and 
(ii) there exists a vertex $v \in C$
such that $\sigma(v) \neq \mu(v)$. 

Assume that 
$C \cap H = \emptyset$.
Then for every doctor $d \in C \cap D$, 
$\sigma(d) = \emptyset$. 
Let $d$ be a doctor in $C \cap D$ such that 
$\sigma(d) \neq \mu(d)$. 
Since $\sigma(d) = \emptyset$, we have 
$\mu(d) \neq \emptyset$ and 
$\mu(d) \succ_{d} \sigma(d)$. 
However, this contradicts the fact that
$\sigma(d) \succsim_d \mu(d)$.
This completes the proof. 
\end{proof} 

Let $\mathcal{C}$ denote the set of matchings $\mu$ in $G$ such that 
there does not exist a non-empty subset of $V$ that strongly blocks $\mu$. 
Let $\mathcal{C}_S$ denote the set of matchings $\mu$ in $G$ such that 
there does not exist a non-empty subset of $V$ that weakly blocks $\mu$. 
Let $\mathcal{C}_{SS}$ denote the set of matchings $\mu$ in $G$ such that 
there does not exist a non-empty subset of $V$ that super weakly blocks $\mu$. 

Bonifacio, Juarez, Neme, and Oviedo~\cite{BonifacioJNO24}
proved that if the matroids are uniform matroids, then 
$\mathcal{S} \subseteq \mathcal{C}$, 
$\mathcal{SS} = \mathcal{C}_S$, and 
$\mathcal{SSS} = \mathcal{C}_{SS}$. 
We prove that 
these results hold in our setting.

\subsection{Useful lemma} 

The following lemma is used to prove Lemma~\ref{lemma:GabowGK74}.

\begin{lemma}[{Brualdi~\cite{Brualdi69}}] \label{lemma:Brualdi69} 
Let ${\bf M}$ be a matroid, and 
let $B,B^{\prime}$ be distinct bases of 
${\bf M}$. 
Then for every element $e \in B \setminus B^{\prime}$, 
there exists an element $f \in B^{\prime} \setminus B$ such that 
$e \in {\sf C}_{\bf M}(f,B)$
and 
$f \in {\sf C}_{\bf M}(e,B^{\prime})$. 
\end{lemma} 

The following lemma plays an important role in the next section. 
For reader's convenience, we give its proof.  

\begin{lemma}[{Gabow, Glover, and Klingman~\cite{GabowGK74}}] \label{lemma:GabowGK74}
Let ${\bf M}$ be a matroid, and 
let $B,B^{\prime}$ be distinct bases of 
${\bf M}$.
Define $k := |B \setminus B^{\prime}|$.
Then there exist
an ordering $(e_1,e_2,\dots,e_k)$ 
of the elements in $B \setminus B^{\prime}$
and 
an ordering $(f_1,f_2,\dots,f_k)$ 
of the elements in $B^{\prime} \setminus B$
satisfying the following conditions.
\begin{itemize}
\item
For every integer $i \in [k]$, 
we have $e_i \in {\sf C}_{\bf M}(f_i,B)$.  
\item 
For every integer $i \in [k]$, 
$(B^{\prime} \setminus \{f_1,f_2,\dots,f_i\}) \cup \{e_1,e_2,\dots,e_i\}$
is a base of ${\bf M}$. 
\end{itemize}
\end{lemma}
\begin{proof}
Lemma~\ref{lemma:Brualdi69} implies that 
there exist an element $e_1 \in B \setminus B^{\prime}$ and 
an element $f_1 \in B^{\prime} \setminus B$ 
such that 
$e_1 \in {\sf C}_{\bf M}(f_1,B)$ 
and $B^{\prime} - f_1 + e_1$ is a base of ${\bf M}$. 
Thus, we assume that $k \ge 2$. 

Let $\ell$ be an integer in $[k-1]$. 
We assume that 
there exist 
elements 
$e_1,e_2,\dots,e_{\ell} \in B \setminus B^{\prime}$
and 
$f_1,f_2,\dots,f_{\ell} \in B^{\prime} \setminus B$
such that the orderings
$(e_1,e_2,\dots,e_{\ell})$
and 
$(f_1,f_2,\dots,f_{\ell})$ satisfy the conditions in this lemma
for every integer $i \in [\ell]$. 
Define 
the base $B^{\circ}$ of ${\bf M}$ by 
\begin{equation*}
B^{\circ} := 
(B^{\prime} \setminus \{f_1,f_2,\dots,f_{\ell}\})
\cup 
\{e_1,e_2,\dots,e_{\ell}\}. 
\end{equation*}
Let $x$ be an element in $B \setminus B^{\circ}$. 
Notice that $x \in B \setminus (B^{\prime} \cup \{e_1,e_2,\dots,e_{\ell}\})$. 
Lemma~\ref{lemma:Brualdi69} implies that 
there exists an element $y \in B^{\circ} \setminus B$ such that 
$x \in {\sf C}_{\bf M}(y,B)$ and 
$y \in {\sf C}_{\bf M}(x,B^{\circ})$. 
Notice that 
$y \in B^{\prime} \setminus (B \cup \{f_1,f_2,\dots,f_{\ell}\})$. 
Thus, 
by defining $e_{\ell+1} := x$ and 
$f_{\ell+1} := y$, we can obtain 
elements 
$e_1,e_2,\dots,e_{\ell+1} \in B \setminus B^{\prime}$
and 
$f_1,f_2,\dots,f_{\ell+1} \in B^{\prime} \setminus B$
satisfying the conditions in this lemma. 
This completes the proof. 
\end{proof} 

\section{Results} 

In this section, we give the results of this paper. 
Our proofs are based on the proofs in \cite{BonifacioJNO24}. 

\begin{theorem}
$\mathcal{S} \subseteq \mathcal{C}$.
\end{theorem}
\begin{proof}
Let $\mu$ be a matching in $\mathcal{S}$. 
Assume that $\mu \notin \mathcal{C}$.
Then 
there exist a non-empty 
subset $C \subseteq V$ and 
a matching $\sigma \in \mathcal{M}(C)$
such that, 
for every vertex $v \in C$, 
$\sigma(v) \succ_v \mu(v)$. 
Notice that since $C$ super weakly blocks $\mu$, 
Lemma~\ref{lemma:non-empty} 
implies that 
$C \cap H \neq \emptyset$.

Let $h$ be a hospital in $C \cap H$. 
If $\sigma(h) \subseteq \mu(h)$, then 
$\mu(h) \succsim_h \sigma(h)$. 
This contradicts the fact that $\sigma(h) \succ_h \mu(h)$. 
Thus, $\sigma(h) \setminus \mu(h) \neq \emptyset$.
If there exists an edge $(d,h) \in \sigma(h) \setminus \mu(h)$ such that 
$\mu(h) + (d,h) \in \mathcal{I}_h$, then 
since $d \in C$ (i.e., $\sigma(d) \succ_d \mu(d)$),
$(d,h)$ strongly blocks $\mu$. 
However, this contradicts the fact that $\mu \in \mathcal{S}$. 
Thus, for every edge $e \in \sigma(h) \setminus \mu(h)$, we have 
$\mu(h) + e \notin \mathcal{I}_h$. 
 
Assume that there exist edges $e = (d,h) \in \sigma(h) \setminus \mu(h)$ and 
$f \in {\sf D}_{{\bf M}_h}(e,\mu(h))$ such that we have 
$e \succ_h f$. 
Then since $d \in C$ (i.e., $\sigma(d) \succ_d \mu(d)$), 
$e$ strongly blocks $\mu$. 
Thus, 
for every edge $e \in \sigma(h) \setminus \mu(h)$ and 
every edge $f \in {\sf D}_{{\bf M}_h}(e,\mu(h))$, we have 
$f \succsim_h e$. 

Define ${\bf M}_h^{\prime} := {\bf M}_h | (\mu(h) \cup \sigma(h))$. 
Since $\mu(h) + e \notin \mathcal{I}_h$ 
for every edge $e \in \sigma(h) \setminus \mu(h)$, 
$\mu(h)$ is a base of ${\bf M}_h^{\prime}$. 
Define $B^{\prime}$ as a base of ${\bf M}_h^{\prime}$ such that 
$\sigma(h) \subseteq B^{\prime}$.
Since $\sigma(h) \in \mathcal{I}_h$, 
(I3) guarantees the existence of such a base $B^{\prime}$ of 
${\bf M}^{\prime}_h$. 
Then we have $B^{\prime} \succsim_h \sigma(h)$.  
Lemma~\ref{lemma:GabowGK74} implies that 
there exist
an ordering $(e_1,e_2,\dots,e_k)$ 
of the elements in $\mu(h) \setminus B^{\prime}$
and 
an ordering $(f_1,f_2,\dots,f_k)$ 
of the elements in $B^{\prime} \setminus \mu(h)$
satisfying the following conditions.
\begin{itemize}
\item
For every integer $i \in [k]$, 
we have $e_i \in {\sf C}_{{\bf M}^{\prime}_h}(f_i,\mu(h))$.  
\item 
For every integer $i \in [k]$, 
$B^{\prime}_i := (B^{\prime} \setminus \{f_1,f_2,\dots,f_i\}) \cup \{e_1,e_2,\dots,e_i\}$
is a base of ${\bf M}^{\prime}_h$. 
\end{itemize}
Notice that, for every integer $i \in [k]$,
we have ${\sf C}_{{\bf M}^{\prime}_h}(f_i,\mu(h)) = {\sf C}_{{\bf M}_h}(f_i,\mu(h))$. 
Furthermore, for every integer $i \in [k]$, 
$B^{\prime}_i \in \mathcal{I}_h$.  

For every integer $i \in [k]$, 
since $f_i \in \sigma(h) \setminus \mu(h)$, 
we have $e_i \succsim_h f_i$. 
Define $B^{\prime}_0 := B^{\prime}$. 
Then for every integer $i \in [k]$, 
we have $B^{\prime}_i = B^{\prime}_{i-1} + e_i - f_i$.
Notice that $B^{\prime}_k = \mu(h)$.  
For 
every integer $i \in [k]$, 
since $e_i \succsim_h f_i$, 
$B^{\prime}_i \succsim_h B^{\prime}_{i-1}$. 
Thus, $\mu(h) = B^{\prime}_k \succsim_h B^{\prime}_0 \succsim_h \sigma(h)$. 
However, this contradicts the fact 
that $\sigma(h) \succ_h \mu(h)$. 
This completes the proof.
\end{proof} 

Notice that \cite[Example~1]{BonifacioJNO24} shows that 
$\mathcal{C} \subseteq \mathcal{S}$ does not necessarily 
hold even when, for every hospital $h \in H$, 
${\bf M}_h$ is a uniform matroid.

\begin{theorem} \label{theorem:ss}
$\mathcal{SS} = \mathcal{C}_{S}$.
\end{theorem}
\begin{proof}
Let $\mu$ be a matching in $\mathcal{C}_S$. 
Assume that $\mu \notin \mathcal{SS}$. 
Then there exists an edge $e = (d,h) \in E \setminus \mu$
that weakly blocks $\mu$. 
If $\mu(h) + e \in \mathcal{I}_h$, then we can prove that 
$V(\mu(h) + e)$ weakly blocks $\mu$ by setting 
$\sigma := \mu + e - \mu(d)$, where
$V(F)$ denotes the set of end vertices of the edges in $F$ for 
each subset $F \subseteq E$. 
Thus, we can assume that 
$\mu(h) + e \notin \mathcal{I}_h$. 
First, we consider the case where 
$e \sim \mu(d)$. 
In this case, since $e$ weakly blocks 
$\mu$, 
there exists an edge $f \in {\sf D}_{{\bf M}_h}(e,\mu(h))$ 
such that 
$e \succ_h f$. 
Notice that 
$\mu(h) + e - f \succ_h \mu(h)$. 
Thus, we can prove that 
$V(\mu(h) + e - f)$ weakly blocks $\mu$ by setting 
$\sigma := \mu + e - f - \mu(d)$. 
Next, we consider the case where 
$e \succ \mu(d)$. 
In this case, since $e$ weakly blocks 
$\mu$, 
there exists an edge $f \in {\sf D}_{{\bf M}_h}(e,\mu(h))$ 
such that 
$e \succsim_h f$. 
Notice that 
$\mu(h) + e - f \succsim_h \mu(h)$. 
Thus, we can prove that 
$V(\mu(h) + e - f)$ weakly blocks $\mu$ by setting 
$\sigma := \mu + e - f - \mu(d)$. 
This is a contradiction. 

Let $\mu$ be a matching in $\mathcal{SS}$. 
Assume that $\mu \notin \mathcal{C}_S$. 
Then there exist a non-empty subset $C \subseteq V$ and 
a matching $\sigma \in \mathcal{M}(C)$
satisfying the following conditions. 
\begin{itemize}
\item
For every vertex $v \in C$, 
$\sigma(v) \succsim_v \mu(v)$. 
\item 
There exists a vertex $v \in C$
such that $\sigma(v) \succ_v \mu(v)$. 
\end{itemize}
Since $C$ super weakly blocks $\mu$, 
Lemma~\ref{lemma:non-empty} 
implies that 
$C \cap H \neq \emptyset$.

First, we consider the case where there exists a hospital $h \in C \cap H$ such that 
$\sigma(h) \succ_h \mu(h)$. 
If $\sigma(h) \subseteq \mu(h)$, then 
$\mu(h) \succsim_h \sigma(h)$. 
However, this contradicts the fact that $\sigma(h) \succ_h \mu(h)$. 
Thus, $\sigma(h) \setminus \mu(h) \neq \emptyset$.
If there exists an edge $(d,h) \in \sigma(h) \setminus \mu(h)$ such that 
$\mu(h) + (d,h) \in \mathcal{I}_h$, then 
since $d \in C$ (i.e., $\sigma(d) \succsim_d \mu(d)$),
$(d,h)$ weakly blocks $\mu$. 
This contradicts the fact that $\mu \in \mathcal{SS}$. 
Thus, for every edge $e \in \sigma(h) \setminus \mu(h)$, we have 
$\mu(h) + e \notin \mathcal{I}_h$. 

Assume that there exist edges $e = (d,h) \in \sigma(h) \setminus \mu(h)$ and 
$f \in {\sf D}_{{\bf M}_h}(e,\mu(h))$ such that  
we have 
$e \succ_h f$. 
Then since $d \in C$ (i.e., $\sigma(d) \succsim_d \mu(d)$), 
$e$ weakly blocks $\mu$. 
Thus, 
for every edge $e \in \sigma(h) \setminus \mu(h)$ and 
every edge $f \in {\sf D}_{{\bf M}_h}(e,\mu(h))$, we have 
$f \succsim_h e$. 

Define ${\bf M}_h^{\prime} := {\bf M}_h | (\mu(h) \cup \sigma(h))$. 
Then since $\mu(h) + e \notin \mathcal{I}_h$ 
for every edge $e \in \sigma(h) \setminus \mu(h)$, 
$\mu(h)$ is a base of ${\bf M}_h^{\prime}$. 
Define $B^{\prime}$ as a base of ${\bf M}_h^{\prime}$ such that 
$\sigma(h) \subseteq B^{\prime}$.
Then we have $B^{\prime} \succsim_h \sigma(h)$.  
Lemma~\ref{lemma:GabowGK74} implies that 
there exist
an ordering $(e_1,e_2,\dots,e_k)$ 
of the elements in $\mu(h) \setminus B^{\prime}$
and 
an ordering $(f_1,f_2,\dots,f_k)$ 
of the elements in $B^{\prime} \setminus \mu(h)$
satisfying the following conditions.
\begin{itemize}
\item
For every integer $i \in [k]$, 
we have $e_i \in {\sf C}_{{\bf M}^{\prime}_h}(f_i,\mu(h)) = {\sf C}_{{\bf M}_h}(f_i,\mu(h))$.  
\item 
For every integer $i \in [k]$, 
$B^{\prime}_i := (B^{\prime} \setminus \{f_1,f_2,\dots,f_i\}) \cup \{e_1,e_2,\dots,e_i\}
\in \mathcal{I}_h$. 
\end{itemize}

For every integer $i \in [k]$, 
since $f_i \in \sigma(h) \setminus \mu(h)$, 
we have $e_i \succsim_h f_i$. 
Define $B^{\prime}_0 := B^{\prime}$. 
Then for every integer $i \in [k]$, 
we have $B^{\prime}_i = B^{\prime}_{i-1} + e_i - f_i$.
Notice that $B^{\prime}_k = \mu(h)$.  
For 
every integer $i \in [k]$, 
since $e_i \succsim_h f_i$, 
$B^{\prime}_i \succsim_h B^{\prime}_{i-1}$. 
Thus, $\mu(h) = B^{\prime}_k \succsim_h B^{\prime}_0 \succsim_h \sigma(h)$. 
However, this contradicts the fact 
that $\sigma(h) \succ_h \mu(h)$. 

Next, we assume that 
$\sigma(h) \sim_h \mu(h)$ for every hospital $h \in C \cap H$. 
In this case, there exists a doctor $d \in C \cap D$ such that 
$\sigma(d) \succ_d \mu(d)$.
Let $g = (d,h)$ denote $\sigma(d)$. 
If there exists an edge $e \in \sigma(h) \setminus \mu(h)$ such that 
$\mu(h) + e \in \mathcal{I}_h$, 
then $e$ weakly blocks $\mu$. 
Thus, 
$\mu(h) + e \notin \mathcal{I}_h$
for every edge $e \in \sigma(h) \setminus \mu(h)$.
If there exists an edge $f \in {\sf D}_{{\bf M}_h}(g,\mu(h))$ such that 
$g \succsim_h f$, then 
$g$ weakly blocks $\mu$. 
Thus, $f \succ_h g$ for every edge $f \in {\sf D}_{{\bf M}_h}(g,\mu(h))$. 
Furthermore, if there exist an edge $e \in \sigma(h) \setminus (\mu(h)+g)$ and 
an edge $f \in {\sf D}_{{\bf M}_h}(e,\mu(h))$ such that 
$e \succ_h f$, then 
$e$ weakly blocks $\mu$. 
Thus, $f \succsim_h e$ for 
every edge $e \in \sigma(h) \setminus (\mu(h)+g)$ and 
every edge $f \in {\sf D}_{{\bf M}_h}(e,\mu(h))$. 

Define ${\bf M}_h^{\prime} := {\bf M}_h | (\mu(h) \cup \sigma(h))$. 
Then since $\mu(h) + e \notin \mathcal{I}_h$
for every edge $e \in \sigma(h) \setminus \mu(h)$, 
$\mu(h)$ is a base of ${\bf M}_h^{\prime}$. 
Define $B^{\prime}$ as a base of ${\bf M}_h^{\prime}$ such that 
$\sigma(h) \subseteq B^{\prime}$.
Then we have $B^{\prime} \succsim_h \sigma(h)$.  
Lemma~\ref{lemma:GabowGK74} implies that 
there exist
an ordering $(e_1,e_2,\dots,e_k)$ 
of the elements in $\mu(h) \setminus B^{\prime}$
and 
an ordering $(f_1,f_2,\dots,f_k)$ 
of the elements in $B^{\prime} \setminus \mu(h)$
satisfying the following conditions.
\begin{itemize}
\item
For every integer $i \in [k]$, 
we have $e_i \in {\sf C}_{{\bf M}^{\prime}_h}(f_i,\mu(h)) = {\sf C}_{{\bf M}_h}(f_i,\mu(h))$.  
\item 
For every integer $i \in [k]$, 
$B^{\prime}_i := (B^{\prime} \setminus \{f_1,f_2,\dots,f_i\}) \cup \{e_1,e_2,\dots,e_i\} \in \mathcal{I}_h$.
\end{itemize}

For every integer $i \in [k]$, 
since $f_i \in \sigma(h) \setminus \mu(h)$, 
we have $e_i \succsim_h f_i$. 
Define $B^{\prime}_0 := B^{\prime}$. 
Then for every integer $i \in [k]$, 
we have $B^{\prime}_i = B^{\prime}_{i-1} + e_i - f_i$.
Notice that $B^{\prime}_k = \mu(h)$.  
For 
every integer $i \in [k]$, 
since $e_i \succsim_h f_i$, 
$B^{\prime}_i \succsim_h B^{\prime}_{i-1}$. 
Furthermore, if $f_i = g$, then 
$e_i \succ_h f_i$ and $B^{\prime}_i \succ_h B^{\prime}_{i-1}$. 
Thus, $\mu(h) = B^{\prime}_k \succ_h B^{\prime}_0 \succsim_h \sigma(h)$. 
However, this contradicts the fact 
that $\sigma(h) \sim_h \mu(h)$. 
\end{proof} 

\begin{theorem} \label{theorem:sss} 
$\mathcal{SSS} = \mathcal{C}_{SS}$.
\end{theorem}
\begin{proof}
Let $\mu$ be a matching in $\mathcal{C}_{SS}$. 
Assume that $\mu \notin \mathcal{SSS}$. 
Then there exists an edge $e = (d,h) \in E \setminus \mu$
that super weakly blocks $\mu$. 
If $\mu(h) + e \in \mathcal{I}_h$, then we can prove that 
$V(\mu(h) + e)$ super weakly blocks $\mu$ by setting 
$\sigma := \mu + e - \mu(d)$.
Assume that 
$\mu(h) + e \notin \mathcal{I}_h$. 
Then since $e$ super weakly blocks 
$\mu$, 
there exists an edge $f \in {\sf D}_{{\bf M}_h}(e,\mu(h))$ 
such that 
$e \succsim_h f$. 
Notice that 
$\mu(h) + e - f \succsim_h \mu(h)$. 
Thus, we can prove that 
$V(\mu(h) + e - f)$ super weakly blocks $\mu$ by setting 
$\sigma := \mu + e - f - \mu(d)$. 
This is a contradiction. 

Let $\mu$ be a matching in $\mathcal{SSS}$. 
Assume that $\mu \notin \mathcal{C}_{SS}$. 
Then there exist a non-empty subset $C \subseteq V$ and 
a matching $\sigma \in \mathcal{M}(C)$
satisfying the following conditions. 
\begin{itemize}
\item
For every vertex $v \in C$, 
$\sigma(v) \succsim_v \mu(v)$. 
\item 
There exists a vertex $v \in C$
such that $\sigma(v) \neq \mu(v)$. 
\end{itemize}
Lemma~\ref{lemma:non-empty} 
implies that 
$C \cap H \neq \emptyset$.

First, we prove that there exists a hospital $h \in C \cap H$ such that 
$\sigma(h) \neq \mu(h)$. 
Let $v$ be a vertex in $C$ such that 
$\sigma(v) \neq \mu(v)$. 
If $v \in H$, then the proof is done. 
Assume that $v \in D$. 
Since $\sigma(v) \succsim_v \mu(v)$, 
$\sigma(v) \neq \emptyset$. 
Let $g = (v,h)$ denote $\sigma(v)$. 
Then $h \in C$. 
Since 
$\sigma(v) \neq \mu(v)$, 
we have $g \in \sigma(h)$ and $g \notin \mu(h)$.
Thus, $\sigma(h) \neq \mu(h)$.

Let $h$ be a hospital in $C \cap H$
such that $\sigma(h) \neq \mu(h)$.
If $\sigma(h) \subsetneq \mu(h)$, then 
$\mu(h) \succ_h \sigma(h)$. 
However, this contradicts the fact that $\sigma(h) \succsim_h \mu(h)$. 
Thus, $\sigma(h) \setminus \mu(h) \neq \emptyset$.
If there exists an edge $(d,h) \in \sigma(h) \setminus \mu(h)$ such that 
$\mu(h) + (d,h) \in \mathcal{I}_h$, then 
since $d \in C$ (i.e., $\sigma(d) \succsim_d \mu(d)$), 
$e$ super weakly blocks $\mu$. 
However, this contradicts the fact that $\mu \in \mathcal{SSS}$.
Thus, for every edge $e \in \sigma(h) \setminus \mu(h)$,  
we have $\mu(h) + e \notin \mathcal{I}_h$. 

Assume that there exist edges $e = (d,h) \in \sigma(h) \setminus \mu(h)$ and 
$f \in {\sf D}_{{\bf M}_h}(e,\mu(h))$ such that we have 
$e \succsim_h f$. 
Then since $d \in C$ (i.e., $\sigma(d) \succsim_d \mu(d)$), 
$e$ super weakly blocks $\mu$. 
Thus, 
for every edge $e = (d,h) \in \sigma(h) \setminus \mu(h)$ and 
every edge $f \in {\sf D}_{{\bf M}_h}(e,\mu(h))$, we have 
$f \succ_h e$. 

Define ${\bf M}_h^{\prime} := {\bf M}_h | (\mu(h) \cup \sigma(h))$. 
Then since $\mu(h) + e \notin \mathcal{I}_h$
for every edge $e \in \sigma(h) \setminus \mu(h)$, 
$\mu(h)$ is a base of ${\bf M}_h^{\prime}$. 
Define $B^{\prime}$ as a base of ${\bf M}_h^{\prime}$ such that 
$\sigma(h) \subseteq B^{\prime}$.
Then we have $B^{\prime} \succsim_h \sigma(h)$.  
Lemma~\ref{lemma:GabowGK74} implies that 
there exist
an ordering $(e_1,e_2,\dots,e_k)$ 
of the elements in $\mu(h) \setminus B^{\prime}$
and 
an ordering $(f_1,f_2,\dots,f_k)$ 
of the elements in $B^{\prime} \setminus \mu(h)$
satisfying the following conditions.
\begin{itemize}
\item
For every integer $i \in [k]$, 
we have $e_i \in {\sf C}_{{\bf M}^{\prime}_h}(f_i,\mu(h)) = {\sf C}_{{\bf M}_h}(f_i,\mu(h))$.  
\item 
For every integer $i \in [k]$, 
$B^{\prime}_i := (B^{\prime} \setminus \{f_1,f_2,\dots,f_i\}) \cup \{e_1,e_2,\dots,e_i\}
\in \mathcal{I}_h$.
\end{itemize}

For 
every integer $i \in [k]$, 
since $f_i \in \sigma(h) \setminus \mu(h)$, 
$e_i \succ_h f_i$. 
Define $B^{\prime}_0 := B^{\prime}$. 
For 
every integer $i \in [k]$, 
since $e_i \succ_h f_i$, 
$B^{\prime}_i \succ_h B^{\prime}_{i-1}$. 
Thus, $\mu(h) = B^{\prime}_k \succ_h B^{\prime}_0 \succsim_h \sigma(h)$. 
However, this contradicts the fact 
that $\sigma(h) \succsim_h \mu(h)$. 
This completes the proof.
\end{proof} 

Finally, we consider algorithmic implications of 
our results. 
Theorems~\ref{theorem:ss} and \ref{theorem:sss} imply that 
if we can determine whether $\mathcal{SS} \neq \emptyset$ and 
$\mathcal{SSS} \neq \emptyset$, then we can determine 
whether $\mathcal{C}_S \neq \emptyset$ and 
$\mathcal{C}_{SS} \neq \emptyset$. 
It is known that, under the assumption that we are given independence 
oracles for the matroids, we can determine 
whether $\mathcal{SS} \neq \emptyset$ and 
$\mathcal{SSS} \neq \emptyset$
in polynomial time~\cite{Kamiyama22,Kamiyama22b}. 
Thus, under the same assumption, 
we can determine 
whether $\mathcal{C}_S \neq \emptyset$ and 
$\mathcal{C}_{SS} \neq \emptyset$ in polynomial time. 

\bibliographystyle{plain}
\bibliography{core_matroid_bib}

\end{document}